\newtheorem{theorem}{Theorem}
\newtheorem{lemma}{Lemma}

\newtheorem{remark}{Remark}
\newtheorem{construction}{Construction}
\newtheorem{example}{Example}

\documentclass[onecolumn,10pt]{IEEEtran}
\renewcommand{\baselinestretch}{1.2}
\usepackage{graphicx,amsmath,amssymb}
\ifCLASSINFOpdf

\else

\fi

\begin{document}

\title{New Optimal Binary Sequences with Period $4p$ via Interleaving Ding-Helleseth-Lam Sequences}

\author{Wei Su, Yang~Yang, and Cuiling Fan
\thanks{W. Su is with School of Economics and Information Engineering, Southwestern University of Finance and Economics, Chengdu, China. Y. Yang and C.L. Fan are with the School of Mathematics, Southwest Jiaotong University, Chengdu, China. Email: suwei@swufe.edu.cn, yang$\_$data@swjtu.edu.cn, fcl@swjtu.edu.cn. }
\thanks{Manuscript received May 28, 2017.}}

\maketitle

\begin{abstract}
Binary sequences with optimal autocorrelation  play
important roles in radar, communication, and cryptography.
Finding new binary sequences with optimal autocorrelation has been an interesting research topic in sequence design.
Ding-Helleseth-Lam sequences are such a class of binary sequences of period $p$, where $p$ is an odd prime with $p\equiv 1(\bmod~4)$. The objective of this letter is to present a construction of binary sequences
of period $4p$ via interleaving four suitable Ding-Helleseth-Lam sequences. This construction generates
new binary sequences with optimal autocorrelation which can not be produced by earlier ones.
\end{abstract}

\begin{IEEEkeywords}
Binary sequences, optimal autocorrelation, interleaving, Ding-Helleseth-Lam sequences.
\end{IEEEkeywords}

\section{Introduction}

Due to simplicity of implementation, binary  sequences with optimal autocorrelation
have important applications in many areas of cryptography, communication and radar.
In cryptography, the sequences can be used to generate key streams in stream cipher encryptions.
In communication and radar, on the other hand, the sequences are employed to acquire the accurate timing information of received signals. During these four decades, searching binary sequences with optimal autocorrelation
has been an interesting research topic in sequence design.
The reader is referred to \cite{GG2005} for more details on binary sequences with optimal autocorrelation
and their applications. See also \cite{CD09}, \cite{TG10} and \cite{TD10} for recent progress on
their constructions.

Given two binary sequences $a=(a(t))$ and $b=(b(t))$ of period $N$, their (periodic) cross-correlation
is defined by
\begin{eqnarray*}
R_{a,b}(\tau)= \sum\limits_{i=0}^{N-1}(-1)^{a(i)+b((i+\tau)_N)}
\end{eqnarray*}
where $a(t),b(t)\in \{0,1\}$ and the addition $(i+\tau)_N$ is the smallest non-negative integer such that $(i+\tau)_N\equiv (i+\tau)(\bmod~N)$.
When the two sequences $a$ and $b$ are identical, the periodic cross-correlation function is
said to be the periodic autocorrelation function, and is denoted by  $R_a$ for short. Furthermore,
these $R_a(\tau), 1\leq \tau\leq N-1$, are referred to as the out-of-phase autocorrelation
values of the sequence $(a(t))$.

Let $a=(a(t))$ be a binary sequence of period $N$ and  $\mathbb{Z}_N=\{0,1,\cdots,N-1\}$ denote the ring of integers modulo $N$. The set
$$
C_a=\{t\in \mathbb{Z}_N: a(t)=1\}
$$
is called the support of $a$, and $a$ is said to be the characteristic sequence of
the set $C_a\subset \mathbb{Z}_N$. It is easy to verify that
\begin{eqnarray}\label{eqn-relation}
R_{a}(\tau)= N-4|(C_a+\tau)\cap C_a|, \tau\in \mathbb{Z}_N.
\end{eqnarray}

It follows from (\ref{eqn-relation}) that $R_{a}(\tau)\equiv N~(\bmod~4)$ for each $1\leq \tau<N$. Accordingly,
in terms of the smallest possible values of the autocorrelation, the optimal values of out-of-phase autocorrelations of binary sequences can be
classified into four types as follows:
\begin{enumerate}
\item [(A)] $R_a(\tau)=0$ for $N\equiv 0 \pmod{4}$;
\item [(B)] $R_a(\tau)\in \{1,-3\}$ for $N\equiv 1 \pmod{4}$;
\item [(C)] $R_a(\tau)\in \{\pm 2\}$ for $N\equiv 2 \pmod{4}$;
\item [(D)] $R_a(\tau)=-1$ for $N\equiv 3 \pmod{4}$.
\end{enumerate}

The sequences in Types (A) and (D) are called perfect sequences and ideal sequences, respectively.
The only known perfect binary sequences up to equivalence is the $(0, 0, 0, 1)$.
It is conjectured that there is no perfect binary sequence of period $N>4$. This conjecture
is widely believed to be true in both mathematical and engineer society. Hence, it is natural to consider the next smallest values for the out-of-phase autocorrelation of a binary sequence of period $N\equiv 0 \pmod{4}$. That is,
$R_a(\tau)\in \{0,\pm 4\}$. If both $4$ and $-4$  occur when $\tau$ rangers from $1$ to $N-1$, then the sequence $a$ is said to be optimal with respect to its correlation magnitude \cite{Yu2008}.

Known constructions of optimal binary sequences of period $N\equiv 0(\bmod~4)$ are summarized as follows.
\begin{enumerate}
\item [1)] $N=q-1$. There were two classes of constructions: The well-known Sidelnikov sequences \cite{Sidelnikov,Lempel} and their slight generalization using $(z+1)^d+az^d+b$  \cite{No}.
\item [2)] $N=4S$, $S$ even.  Recently, Krengel and Ivanov \cite{Krengel2016} proposed two constructions of  optimal binary sequences of period $4S$. Their constructions are based on almost perfect binary sequences of length $2S$ given by Wolfmann \cite{Wolfmann}, and optimal binary sequences of length $S\equiv 2(\bmod~4)$ (i.e., the Sidelnikov sequences \cite{Sidelnikov,Lempel} or the Ding-Helleseth-Martinsen sequences \cite{DHM}).
\item [3)] $N=4S$, $S$ odd. Arasu, Ding, Helleseth, Kumar, and Martisen \cite{Arasu01} proposed optimal binary sequences of length $4S$ from an almost difference set. This was respectively generated by Zhang, Lei, and Zhang \cite{Zhang2006} for the case $S\equiv 3(\bmod~4)$ being an odd prime, and by Yu and Gong \cite{Yu2008} based on a perfect sequence of period $4$ and an ideal sequence of period $S$, where $S=2^n-1$, $S=p$ where $p\equiv 3(\bmod~4)$, or $S=p(p+2)$, where $p$ and $p+2$ are twin primes. In \cite{Yu2008}, Yu and Gong also constructed binary sequences of period $4(2^{2k}-1)$ with out-of-phase auto-correlation in $\{0,\pm 4\}$. In 2010, Tang and Gong \cite{TG10} gave three new constructions for optimal binary sequences of period $4S$ by using interleaving method, whose  columns sequences are the three types of pairs of sequences: i) generalized GMW sequence pair of period $S=2^{2k}-1$, where $k$ is a positive integer; ii) twin-prime sequence pair of period $S=p(p+2)$, where $p$ and $p+2$ are twin primes; iii) Legendre sequence pair of period $S=p$, where $p$ is an odd prime. Those sequences have optimal auto-correlation $R_a(\tau)\in \{0,\pm 4\}$ for all $1\le\tau<4S$. Recently, choosing arbitrary two ideal binary sequences of the same length, Tang and Ding \cite{TD10} constructed new classes of optimal binary sequences via interleaving method firstly introduced by Gong \cite{Gong1995}, which is a useful method to construct sequences with low out-of-phase auto-correlation and cross-correlation (This will be introduced in the next section).
\end{enumerate}


Ding-Helleseth-Lam sequences are such a class of binary sequences of period $p$, where $p$ is an odd prime with $p\equiv 1(\bmod~4)$.
The objective of this letter is to present a construction of binary sequences
of period $4p$ via interleaving four suitable Ding-Helleseth-Lam sequences.
It will be seen later that our construction generates
new binary sequences with optimal autocorrelation which can not be produced by earlier ones.

The rest of this paper is organized as follows. In Section II, we  recall the interleaving method, Ding-Helleseth-Lam sequences \cite{DHL} and their correlation properties \cite{Su}. In Section III, we  present eight classes of new interleaved sequences by choosing suitable four Ding-Helleseth-Lam sequences as column sequences. Those new sequences have optimal auto-correlation magnitude. Finally, we conclude this letter.

\section{Preliminaries}

In this section, we  give an introduction to interleaved technique and Ding-Helleseth-Lam
sequences which will be used to construct new optimal binary sequences in the sequel.

\subsection{Interleaved Technique}
Interleaved method proposed by Gong \cite{Gong1995} is a powerful technique in sequence design.
The key idea of this method is to obtain long sequences with good correlation  from
shorter ones. Following the notation and terminology in \cite{Gong1995}, we give a shot introduction
to this method.  Let  $a_k=(a_k(0),a_k(1),\cdots,a_k(N-1))$ be a sequence of period $N$, where $0\leq k\leq M-1$. From these $M$ sequences, we can obtain an $N\times M$ matrix $U=(U_{i,j})$:
\begin{eqnarray*}
U=\left(
      \begin{array}{cccc}
        a_0(0) & a_1(0) & \cdots & a_{M-1}(0) \\
        a_0(1) & a_1(1) & \cdots & a_{M-1}(1) \\
        \vdots & \vdots & \ddots & \vdots \\
        a_0(N-1) & a_1(N-1) & \cdots & a_{M-1}(N-1) \\
      \end{array}
    \right).
\end{eqnarray*}
Concatenating the successive rows of the matrix above, an interleaved sequence $u=(u(t))$ of period $MN$ is
defined by
$$u_{iM+j}=U_{i,j},\,\,\,\,0\le i<N,0\le j<M.$$
For convenience, we denote $u$ by
$$
u=I(a_0,a_1,\cdots,a_{M-1}),$$
where $I$ is called the interleaving operator. Herein and hereafter $a_0,a_1,\cdots,a_{M-1}$ are called the {column sequences} of $u$.

Let $L$ be  the (left  cyclical) shift operator of any vector, i.e.,  $L(c)=(c(1),c(2),\cdots,c(N-1), c(0))$ for any $c=(c(0),c(1),\cdots,c(N-1))$. Then $L^{\tau}(u)$ can be represented as
$$L^{\tau}(u)=I(L^{\tau_1}(a_{\tau_2}),
\cdots,L^{\tau_1}(a_{M-1}),L^{\tau_1+1}(a_0),\cdots,L^{\tau_1+1}(a_{\tau_2-1})).$$
where $\tau=\tau_1M+\tau_2$ $(0\le \tau_1<N, 0\le \tau_2<M)$. It is easy to verify that the
periodic autocorrelation of $u$ at shift $\tau$ is given by
\begin{eqnarray*}\label{eq interleaving}
R_u(\tau)
&=&\sum\limits_{k=0}^{M-\tau_2-1}R_{a_k,a_{k+\tau_2}}(\tau_1)
+\sum\limits_{k=M-\tau_2}^{M-1}R_{a_k,a_{k+\tau_2-M}}(\tau_1+1).
\end{eqnarray*}
This means that the autocorrelation of $u$ is fully determined by the autocorrelation and crosscorrelation of column sequences $a_i$.

\subsection{Ding-Helleseth-Lam sequences}

Let  $p=4f+1$ is an odd prime, where $f$ is a positive integer. Let $\alpha$ be a generator of the multiplicative group of the residue ring $\mathbb{Z}_p$, and let $D_i=\{\alpha^{i+4j}: 0\le j<f \}$, $0\le i<4$. Those $D_i, 0\le i<4$, are called the cyclotomic classes of order $4$ with respect to $\mathbb{Z}_p$.

In \cite{DHL}, Ding, Helleseth, and Lam constructed optimal binary sequences of odd prime period $p$ by using cyclotomic number of order $4$.

\begin{lemma}[Ding-Helleseth-Lam sequences, \cite{DHL}]\label{le 1}
Let $p=4f+1=x^2+4y^2$ be an odd prime, where $f,x,y$ are integers. Let $D_0,D_1,D_2,D_3$ be  the cyclotomic classes of order $4$ with respect to $\mathbb{Z}_p$. Assume that $s_1, s_2, s_3, s_4$ are binary sequences of period $p$ with supports $D_0\cup D_1$, $D_0\cup D_3$, $D_1\cup D_2$ and $D_2\cup D_3$, respectively. Then each $s_i$ is optimal, i.e., $R_{s_i}(\tau)\in \{1,-3\}$ for all $1\le\tau<p$, if and only if $f$ is odd and $y=\pm 1$.
\end{lemma}

The correlation values of Ding-Helleseth-Lam sequences $s_1, s_2, s_3, s_4$ have been determined in \cite{Su}, which are useful for the main result of this paper. Here we list it as follows.

\begin{lemma}\label{le odd}
Let $s_1, s_2, s_3, s_4$ be the Ding-Helleseth-Lam sequences in Lemma \ref{le 1}. For odd $f$, the autocorrelation and cross-correlation of $s_1,s_2,s_3,s_4$  are given in Table \ref{tab odd f}.

\begin{table}\caption{The autocorrelation and cross-correlation of  Ding-Helleseth-Lam sequences}\label{tab odd f}
\begin{center}
\begin{tabular}{|c|c|c|c|c|c|}
  \hline
$\tau$ & $\{0\}$ & $D_0$ & $D_1$ & $D_2$ & $D_3$ \\ \hline
$R_{s_1}(\tau)$ & $p$ & $-2y-1$ & $2y-1$ & $-2y-1$ & $2y-1$\\ \hline
$R_{s_2}(\tau)$ & $p$ & $2y-1$ & $-2y-1$ & $2y-1$ & $-2y-1$\\ \hline
$R_{s_3}(\tau)$ & $p$ & $2y-1$ & $-2y-1$ & $2y-1$ & $-2y-1$\\ \hline
$R_{s_4}(\tau)$ & $p$ & $-2y-1$ & $2y-1$ & $-2y-1$ & $2y-1$\\  \hline
$R_{s_1,s_2}(\tau)$ & $1$ & $x$ & $-x+2$ & $x$ & $-x-2$\\  \hline
$R_{s_1,s_3}(\tau)$ & $1$ & $-x+2$ & $x$ & $-x-2$ & $x$\\  \hline
$R_{s_1,s_4}(\tau)$ & $2-p$ & $2y+3$  & $3-2y$ & $2y-1$ & $-1-2y$ \\ \hline
$R_{s_2,s_1}(\tau)$ & $1$ & $x$ & $-x-2$ & $x$ & $-x+2$ \\  \hline
$R_{s_2,s_3}(\tau)$ & $2-p$ & $3-2y$  & $2y-1$ & $-1-2y$ & $3+2y$ \\ \hline
$R_{s_2,s_4}(\tau)$ & $1$ & $-x+2$  & $x$ & $-x-2$ & $x$ \\ \hline
$R_{s_3,s_1}(\tau)$ & $1$  & $-x-2$ & $x$ & $-x+2$ & $x$\\  \hline
$R_{s_3,s_2}(\tau)$ & $2-p$  & $-1-2y$ & $3+2y$ & $3-2y$  & $2y-1$\\ \hline
$R_{s_3,s_4}(\tau)$ & $1$ & $x$ & $-x+2$  & $x$ & $-x-2$\\ \hline
$R_{s_4,s_1}(\tau)$ & $2-p$ & $2y-1$ & $-1-2y$  & $2y+3$  & $3-2y$\\ \hline
$R_{s_4,s_2}(\tau)$ & $1$ & $-x-2$ & $x$ & $-x+2$  & $x$\\ \hline
$R_{s_4,s_3}(\tau)$ & $1$ & $x$ & $-x-2$& $x$ & $-x+2$ \\ \hline
\end{tabular}
\end{center}
\end{table}
\end{lemma}

\section{New Optimal Binary Sequences with Period $4p$ via Interleaving Ding-Helleseth-Lam Sequences}
In this section, we construct new optimal binary sequences via interleaved technique and Ding-Helleseth-Lam Sequences. From now on, we always suppose that $p=4f+1=x^2+4y^2$ is an odd prime, where $x$ is an integer, $y=\pm 1$, and $f$ is an odd integer. Let  $s_1, s_2, s_3$ and $s_4$ be the  Ding-Helleseth-Lam Sequences in Lemma \ref{le 1}.

We first propose a generic simple construction of binary sequences with period $4p$ based on interleaved technique and Ding-Helleseth-Lam Sequences.

\begin{construction}\label{construction-main}
Let $a_0,a_1,a_2,a_3$ be four binary sequences of length $p$ and $b=(b(0),b(1),b(2),b(3))$ be a binary sequence of length $4$. Construct a binary sequence $u=(u(t))$ of length $4p$  as follows:
\begin{eqnarray}\label{eq u}
u=I(a_0+b(0),L^{d}(a_1)+b(1),L^{2d}(a_2)+b(2),L^{3d}(a_3)+b(3)),
\end{eqnarray}
where  $d$ is some integer with $4d\equiv 1~(\bmod~p)$.
\end{construction}

\begin{remark}For Construction \ref{construction-main}, we have the following comments.
\begin{itemize}
\item[1.] When $b=(b(0),b(1),b(2),b(3))\in \{(0,0,0,1), (1,1,1,0)\}$, and
$a_i, 0\le i\le 3$ are chosen form the first type and the second type Legendre sequences of period $p$, the sequence $u$ generated by  Construction \ref{construction-main} is exactly the binary sequence with optimal correlation reported in \cite{TG10}.

\item[2.] The following results show that the resultant sequence $u$ by Construction \ref{construction-main} also has optimal autocorrelation if the column sequences $a_0,a_1,a_2$ and $a_3$ are properly chosen
from the Ding-Helleseth-Lam sequences, and the binary sequence $b=(b(0),b(1),b(2),b(3))$ satisfies $b(0)=b(2)$ and $b(1)=b(3)$. Therefore, our construction can generate new optimal binary sequences which cannot produced by known ones.
\end{itemize}

\end{remark}

\begin{theorem}\label{th 4p}
Let $b=(b(0),b(1),b(2),b(3))$ be a binary sequence with $b(0)=b(2)$ and $b(1)=b(3)$, and $(a_0,a_1,a_2,a_3)=(s_3,s_2,s_1,s_1)$.
Then the binary sequence $u$ by Construction \ref{construction-main} is optimal.
\end{theorem}

\begin{proof} For any $\tau$, $1\le\tau<4p$, we can write $\tau=4\tau_1+\tau_2$, where ($0\le\tau_1<p$ and $0<\tau_2<4$) or ($0<\tau_1<p$ and $\tau_2=0$). Consider the auto-correlation of $u$ in four cases according to $\tau_2=0,1,2,3$:
\begin{enumerate}
\item $\tau_2=0$: In this case, one has $0<\tau_1<p$ and
$$L^{\tau}(u)=I(L^{\tau_1}(s_3)+b(0),L^{\tau_1+d}(s_2)+b(1),L^{\tau_1+2d}(s_1)+b(2),L^{\tau_1+3d}(s_1)+b(3)).$$
Then the auto-correlation of $u$ at shift $\tau$ is equal to
\begin{eqnarray*}
R_{u}(\tau)&=&
R_{s_3}(\tau_1)+R_{s_2}(\tau_1)+2R_{s_1}(\tau_1)\\
&=&\left\{\begin{array}{ll}
4p, & \tau_1=0\\
-4, & \tau_1\ne 0.
\end{array}\right.
\end{eqnarray*}
where the last equal sign is due to the auto-correlation of $s_1$, $s_2$ and $s_3$ given by Lemma \ref{le odd}.
\item $\tau_2=1$: In this case, one has $0\le \tau_1<p$ and $$L^{\tau}(u)=I(L^{\tau_1+d}(s_2)+b(1),L^{\tau_1+2d}(s_1)+b(2),L^{\tau_1+3d}(s_1)+b(3),L^{\tau_1+1}(s_3)+b(0)).$$
Then the auto-correlation of $u$ at shift $\tau$ is equal to
\begin{eqnarray*}
R_{u}(\tau)&=&(-1)^{b(0)+b(1)}R_{s_3,s_2}((\tau_1+d)_p)+(-1)^{b(1)+b(2)}R_{s_2,s_1}((\tau_1+d)_p)\\
&&+(-1)^{b(2)+b(3)}R_{s_1}((\tau_1+d)_p)+(-1)^{b(3)+b(0)}R_{s_1,s_3}((\tau_1+1-3d)_p)\\
&=&(-1)^{b(0)+b(1)}R_{s_3,s_2}((\tau_1+d)_p)+(-1)^{b(1)+b(2)}R_{s_2,s_1}((\tau_1+d)_p)\\
&&+(-1)^{b(2)+b(3)}R_{s_1}((\tau_1+d)_p)+(-1)^{b(3)+b(0)}R_{s_1,s_3}((\tau_1+d)_p)\\
&=&\left\{\begin{array}{ll}
4(-1)^{b(0)+b(1)}, & (\tau_1+d)_p=0\\
4y(-1)^{b(0)+b(1)}, & (\tau_1+d)_p\in D_0\cup D_2\\
-4y(-1)^{b(0)+b(1)}, & (\tau_1+d)_p\in D_1\cup D_3
\end{array}\right.
\end{eqnarray*}
where the second equality is due to $(\tau_1+1-3d)_p= (\tau_1+d)_p$, and the last equal sign is due to the correlation of $s_1$, $s_2$ and $s_3$ given by Lemma \ref{le odd}.
\item $\tau_2=2$:
In this case, one has $0\le \tau_1<p$ and $$L^{\tau}(u)=I(L^{\tau_1+2d}(s_1)+b(2),L^{\tau_1+3d}(s_1)+b(3),L^{\tau_1+1}(s_3)+b(0),L^{\tau_1+d+1}(s_2)+b(1)).$$
Then by Lemma \ref{le odd},  the auto-correlation of $u$ at shift $\tau$ is equal to
\begin{eqnarray*}
R_{u}(\tau)&=&(-1)^{b(0)+b(2)}R_{s_3,s_1}((\tau+2d)_p)+(-1)^{b(1)+b(3)}R_{s_2,s_1}((\tau_1+2d)_p)\\
&&+(-1)^{b(2)+b(0)}R_{s_1,s_3}((\tau_1+1-2d)_p)+(-1)^{b(3)+b(1)}R_{s_1,s_2}((\tau_1+1-2d)_p)\\
&=&(-1)^{b(0)+b(2)}R_{s_3,s_1}((\tau+2d)_p)+(-1)^{b(1)+b(3)}R_{s_2,s_1}((\tau_1+2d)_p)\\
&&+(-1)^{b(2)+b(0)}R_{s_1,s_3}((\tau_1+2d)_p)+(-1)^{b(3)+b(1)}R_{s_1,s_2}((\tau_1+2d)_p)\\
&=&\left\{\begin{array}{ll}
4, & (\tau_1+2d)_p=0\\
0, & (\tau_1+2d)_p\ne 0
\end{array}\right.
\end{eqnarray*}
where the second equality is due to $(\tau_1+1-2d)_p= (\tau_1+2d)_p$, and the last equal sign is due to the correlation of $s_1$, $s_2$ and $s_3$ given by Lemma \ref{le odd}.
\item $\tau_2=3$:
In this case, one has $0\le \tau_1<p$ and $$L^{\tau}(u)=I(L^{\tau_1+3d}(s_1)+b(3),L^{\tau_1+1}(s_3)+b(0),L^{\tau_1+d+1}(s_2)+b(1),L^{\tau_1+2d+1}(s_1)+b(2)).$$
Then by Lemma \ref{le odd},  the auto-correlation of $u$ at shift $\tau$ is equal to
\begin{eqnarray*}
R_{u}(\tau)&=&(-1)^{b(0)+b(3)}R_{s_3,s_1}((\tau_1+3d)_p)+(-1)^{b(1)+b(0)}R_{s_2,s_3}((\tau_1+1-d)_p)\\
&&+(-1)^{b(2)+b(1)}R_{s_1,s_2}((\tau_1+1-d)_p)+(-1)^{b(3)+b(2)}R_{s_1}((\tau_1+1-d)_p)\\
&=&(-1)^{b(0)+b(3)}R_{s_3,s_1}((\tau_1+3d)_p)+(-1)^{b(1)+b(0)}R_{s_2,s_3}((\tau_1+3d)_p)\\
&&+(-1)^{b(2)+b(1)}R_{s_1,s_2}((\tau_1+3d)_p)+(-1)^{b(3)+b(2)}R_{s_1}((\tau_1+3d)_p)\\
&=&\left\{\begin{array}{ll}
4(-1)^{b(0)+b(1)}, & (\tau_1+3d)_p=0\\
-4y(-1)^{b(0)+b(1)}, & (\tau_1+3d)_p\in D_0\cup D_2\\
4y(-1)^{b(0)+b(1)}, & (\tau_1+3d)_p\in D_1\cup D_3
\end{array}\right.
\end{eqnarray*}
where the second equality is due to $(\tau_1+1-d)_p= (\tau_1+3d)_p$, and the last one is due to the correlation of $s_1$, $s_2$ and $s_3$ given by Lemma \ref{le odd}.
\end{enumerate}

According to the discussion above, we have $R_u(\tau)\in \{0,\pm 4\}$ for all $1\le\tau<4p$ which means that $u$ has optimal autocorrelation. The proof of this theorem is completed.
\end{proof}

\begin{theorem}
Let $b=(b(0),b(1),b(2),b(3))$ be a binary sequence with $b(0)=b(2)$ and $b(1)=b(3)$, and $(a_0,a_1,a_2,a_3)$ be chosen from
$$
\{
(s_2,s_3,s_1,s_1),(s_4,s_1,s_2,s_2),(s_1,s_4,s_2,s_2),(s_4,s_1,s_3,s_3),
(s_1,s_4,s_3,s_3),(s_3,s_2,s_4,s_4),(s_2,s_3,s_4,s_4)\}.
$$
Then the binary sequence $u$ by Construction \ref{construction-main} is optimal.
\end{theorem}

\begin{proof}
The proof is similar to that of Theorem \ref{th 4p}, and thus is omitted here.
\end{proof}


Finally, we conclude this section by giving an example to illustrate our construction.

\begin{example}
Let $p=29$, and $\alpha=2$ be a primitive element of the residue ring $\mathbb{Z}_{p}$. Then \begin{eqnarray*}
D_0&=&\{ 1, 7, 16, 20, 23, 24, 25 \},\\
D_1&=&\{ 2, 3, 11, 14, 17, 19, 21 \},\\
D_2&=&\{ 4, 5, 6, 9, 13, 22, 28 \},\\
D_3&=&\{ 8, 10, 12, 15, 18, 26, 27 \}
\end{eqnarray*}
are four cyclotimic classes of order $4$ with respect to $\mathbb{Z}_p$. In this case, $x=5$, $y=-1$, and $f=7$. Generate three Ding-Helleseth-Lam sequences with supports $D_0\cup D_1$, $D_0\cup D_3$, $D_1\cup D_2$, i.e.,
\begin{eqnarray*}
s_1&=&(0, 1, 1, 1,  0,  0,  0, 1,  0,  0,  0, 1,  0,  0, 1,  0, 1, 1,  0, 1, 1, 1,  0, 1, 1, 1,  0,  0,  0)\\
s_2&=&(0, 1,  0,  0,  0,  0,  0, 1, 1,  0, 1,  0, 1,  0,  0, 1, 1,  0, 1,  0, 1,  0,  0, 1, 1, 1, 1, 1,  0)\\
s_3&=&(0,  0, 1, 1, 1, 1, 1,  0,  0, 1,  0, 1,  0, 1, 1,  0,  0, 1,  0, 1,  0, 1, 1,  0,  0,  0,  0,  0, 1).
\end{eqnarray*}
Let $b=(0,0,0,0)$, $(a_0,a_1,a_2,a_3)=(s_3,s_2,s_1,s_1)$, and  $d=22$. By (\ref{eq u}), we have the interleaved sequence:
\begin{eqnarray*}
u&=&I(s_3,L^d(s_2),L^{2d}(s_1),L^{3d}(s_1))\\
&=&(0, 1, 0, 1, 0, 0, 1, 1, 1, 0, 1, 1, 1, 0, 0, 0, 1, 0, 1, 1, 1, 0, 1, 1, 1, 1, 1, 0, 0, 1, 0,1,\\&& 0, 0, 1, 0, 1, 1, 1, 0, 0, 1, 1, 1, 1, 1, 0, 0, 0, 1, 0, 0, 1, 1, 0, 0, 1, 0, 0, 1, 0, 0, 1, 0, 0, \\&& 1, 1, 0, 1, 0, 1, 0, 0, 1, 0, 1, 1, 0, 0, 1, 0, 1, 0, 1, 1, 1, 1, 1, 1, 0, 0, 0, 0, 0, 0, 0, 0, 1, \\&& 0,0, 0, 0, 1, 1, 0, 1, 0, 1, 0, 0, 0, 1, 1, 1, 1, 0).
\end{eqnarray*}
By computer experiment, the auto-correlation of $u$ is given by
\begin{eqnarray*}
\{R_u(\tau)\}_{\tau=1}^{115}&=&
\{-4, 0, 4, -4, -4, 0, -4, -4, -4, 0, 4, -4, -4, 0, 4, -4, 4, 0, 4, -4, 4, 0, -4, -4, -4, \\&& 0, 4,
-4, -4, 0, 4, -4, -4, 0, -4, -4, 4, 0, 4, -4, 4, 0, 4, -4, -4, 0, 4, -4, -4, 0, -4, \\&&-4, -4, 0, 4,
-4, -4, 4, -4, -4, 4, 0, -4, -4, -4, 0, -4, -4, 4, 0, -4, -4, 4, 0, 4, -4,
\\&& 4, 0, 4, -4, -4, 0, -4,
-4, 4, 0, -4, -4, 4, 0, -4, -4, -4, 0, 4, -4, 4, 0, 4, -4, 4, 0, -4, \\&&-4, 4, 0, -4, -4, -4, 0, -4, -4, 4, 0, -4\}.
\end{eqnarray*}
Hence $R_u(\tau)\in \{0,\pm 4\}$ for all $1\le\tau<116$.
\end{example}

\section{Conclusion}

In this letter, we proposed a construction of binary sequences of period $4p$ with  the interleaved structure
$$u=I(a_0+b(0),L^{d}(a_1)+b(1),L^{2d}(a_2)+b(2),L^{3d}(a_3)+b(3)).$$
where $d$ is some integer with $4d\equiv 1~(\bmod~ p)$ and the column sequences $a_i, 0\leq i\leq 3$ are appropriately selected from the Ding-Helleseth-Lam sequences. Our construction contains one earlier construction of binary optimal sequences as  special cases, and can produce new binary sequences with optimal autocorrelation. It may be possible and interesting to find other column sequences to obtain more optimal binary sequences using this interleaved structure.

\end{document}